\documentclass[smallabstract,smallcaptions]{dccpaper}

\usepackage{epsfig}
\usepackage{citesort}
\usepackage{amsmath}
\usepackage{amssymb}
\usepackage{color}
\usepackage{url}
\usepackage{sansmath}
\usepackage{amsthm}
\usepackage{mymacros}
\usepackage{url}
\usepackage{authblk}

\usepackage{booktabs}

\newlength{\figurewidth}
\newlength{\smallfigurewidth}

\newtheorem{lemma}{Lemma}

\newtheorem{definition}{Definition}
\DeclareMathAlphabet{\mathqhv}{OT1}{qhv}{m}{n}

\begin{document}

\setlength{\affilsep}{0.5em}

\title{%
\large \textbf{A grammar compressor for collections of reads with applications to the construction of the BWT}
}

\author[ ]{Diego D\'iaz-Dom\'inguez}
\author[ ]{Gonzalo Navarro}
\affil[ ]{CeBiB --- Center for Biotechnology and Bioengineering, Chile}
\affil[ ]{Department of Computer Science, University of Chile, Santiago, Chile}
\affil[ ]{\texttt{\{diediaz,gnavarro\}@dcc.uchile.cl}}



\maketitle
\thispagestyle{empty}

\begin{abstract}
We describe a grammar for DNA sequencing reads from which we can compute the BWT directly. Our motivation is to perform in succinct space genomic analyses that require complex string queries not yet supported by repetition-based self-indexes. Our approach is to store the set of reads as a grammar, but when required, compute its BWT to carry out the analysis by using self-indexes. Our experiments in real data showed that the space reduction we achieve with our compressor is competitive with LZ-based methods and better than entropy-based approaches. Compared to other popular grammars, in this kind of data, we achieve, on average, 12\% of extra compression and require less working space and time. 
\end{abstract}

\Section{Introduction}

\emph{Grammar compression} of a text $T[1..n]$ consists in building a context-free grammar $\mathcal G$ that generates (only) $T$ \cite{KY00}. Efficient grammar constructions, like RePair~\cite{l2000off}, achieve very good compression ratios in practice. When $T$ is a repetitive text collection, in particular, the produced grammar $\mathcal G$ can be much smaller than $T$, even breaking the statistical entropy lower bound \cite{k2013com}.

One of the benefits of this technique compared to other compression schemes that succeed on repetitive texts, like Lempel-Ziv, is that we can directly access any substring of $T$ from $\mathcal G$ with only an additive logarithmic time penalty \cite{b2015ran},
thereby enabling manipulation of the data always in compressed form. Further, it is possible to develop self-indexes of size $O(|\mathcal{G}|)$, which support indexed string searches \cite{cl2011s}. 

Still, this functionality is limited compared to the complex sequence analyses required in computational biology scenarios \cite{makinen2015genome}. Many of those problems, for example computing maximal repeats, maximal unique matches~\cite{b2013ver}, or suffix-prefix overlaps~\cite{v2010app}, rely on the Burrows-Wheeler Transform (BWT) \cite{bw94}, a permutation of $T$ that can be compressed significantly \cite{ferragina2005indexing}, even for highly repetitive text collections \cite{ma2010s,g2018op}.

The so-called Run-Length BWT (RLBWT) \cite{ma2010s,g2018op} exploits the fact that, on highly repetitive text collections, the BWT consists of a small number of long runs of the same letter. It can then enable complex sequence analyses on genome collections in very little space. Still, this is not the most common type of sequence collection one finds in bioinformatic applications. Genomes are reconstructed from huge multisets of short and overlapping DNA strings called sequencing reads. Assembling genomes is expensive as it requires extensive computations of suffix-prefix overlaps between the reads, or aligning them to a reference genome.  As a consequence,  those large sets of reads, which are the form in which sequencing technologies deliver their output, are also the most common form in which sequencing data is available, much more common than assembled genomes. 

On these sets of short sequences, the RLBWT does not compress significantly \cite{dolle2017using}, whereas grammars and Lempel-Ziv still obtain good space reductions; in particular, grammars permit manipulating the set of reads directly in compressed form.

Some authors have tried to implement regular bioinformatic analyses on top of the BWT of the reads \cite{simp2012eff,cox2012large,dolle2017using}, as this representation retains more information and uses less space than typical plain genomic-tailored data structures. The problem, however, is that decompressing the collection and then building its BWT requires significant storage and processing resources. An interesting alternative would be to build the transform directly from the compressed data. As far as we know, this idea has been implemented only from Lempel-Ziv compression and is considerably slow \cite{pol2017from}. As discussed, the Lempel-Ziv format does not enable, on the other hand, direct access to the reads for other purposes.

\paragraph{Our contribution.} We propose a new grammar aimed for collections of reads that (i) compresses them at high ratios, (ii) provides fast direct access to the reads in compressed form, and (iii) efficiently computes the BWT for string multisets~\cite{cox2012large} directly from the representation. The working space of our compressor is 50\%--60\% of the input, which is far less than most grammar construction algorithms. On top of the grammar, we devise an algorithm for building the BWT that requires space proportional to the number of rules plus the number of runs in the BWT.

\Section{Related concepts}

\paragraph{Suffix Array and Burrows-Wheeler Transform (BWT).} 

Consider a string $T[1..n-1]$ over alphabet $\Sigma[2..\sigma]$, and the sentinel symbol $\Sigma[1]=\texttt{\$}$, which we append at the end of $T$. The \emph{suffix array} (SA)~\cite{MM93} of $T$ is a permutation of $[n]$ that enumerates the suffixes $T[i..n]$ of $T$ in increasing lexicographic order, $T[SA[i]..n] < T[SA[i+1]..n]$. 
The BWT~\cite{bw94} is a permutation of the symbols of $T$ obtained by extracting the symbol that precedes each suffix in $SA$, that is, $BWT[i] = T[SA[i]-1]$ (assuming $T[0]=T[n]=\texttt{\$}$). 
A run-length compressed representation of the BWT
\cite{ma2010s} adds sublinear-size structures that compute, in logarithmic time, the so-called $\mathsf{LF}$ step and its inverse: if $BWT[j]$ corresponds to $T[i]$ and $BWT[j']$ to $T[i-1]$ (or to $T[n]=\textsf{\$}$ if $i=1$), then $\mathsf{LF}(j)=j'$ and $\mathsf{LF}^{-1}(j')=j$. Note that $\mathsf{LF}$ regards $T$ as a circular string.

Let $\mathcal{T}=\{T_1,T_2,...T_m\}$ be a collection of $m$ strings of average size $k$. We then define the string $T[1..n]=T_1\texttt{\$}T_2\texttt{\$}..T_n\texttt{\$}$. The extended BWT (eBWT) of $\mathcal{T}$ \cite{cox2012large} regards it as a set of independent circular strings: the BWT of $T$ is slightly modified so that, if $eBWT[j]$ corresponds to $T_i[1]$ inside $T$, then $\mathsf{LF}(j)=j'$, so that $eBWT[j']$ corresponds to the sentinel \texttt{\$} at the end of $T_i$, not of $T_{i-1}$.

\paragraph{Induced suffix sorting (ISS).} ISS~\cite{n2009li} is a technique that computes the lexicographical ranks of a subset of suffixes in a string $T$ and then it uses the result to induce the order of the rest. This method is the underlying procedure in several algorithms that build the SA~\cite{n2009li} and the BWT~\cite{ok2009li} in linear time. For this article, the part of the ISS algorithm that computes the lexicographical ranks of a subset of suffixes is of interest. The authors give the following definitions:

\begin{definition}
A character $T[i]$ is called L-type if $T[i]>T[i-1]$ or if $T[i]=T[i-1]$ and $T[i-1]$ also L-type. On the other hand, $T[i]$ is said to be S-type if $T[i]<T[i]$ or if $T[i]=T[i+1]$ and $T[i+1]$ is also S-type. By default, symbol $T[n]$, the one with the sentinel, is S-type.  
\end{definition}

\begin{definition}
$T[i]$ is called LMS-type if $T[i]$ is S-type and $T[i-1]$ is L-type.
\end{definition}

\begin{definition}
A LMS substring is (i) a substring $T[i..j]$ with both $T[i]$ and $T[j]$ being LMS characters, and there is no other LMS character in the substring, for $i \neq j$; or (ii) the sentinel itself.
\end{definition}

The algorithm only computes the ranks of the suffixes prefixed by $LMS$ substrings. It obtains the ranks by sorting the substrings lexicographically. When an $LMS$ substring is prefix of another, the smallest one gets the greatest rank \cite{n2009li}. If there are at least two $LMS$ substrings with the same sequence in $T$, the algorithm replaces all the substrings by their orders and applies recursively the same idea until all the characters in $T$ are distinct. 

Recently, Nunes \emph{et al.}~\cite{n2018gr} showed that this procedure can be used to build a grammar of the text. In every recursive step of ISS, they get the set of distinct $LMS$ substring to create new rules. The ranks of the strings in the set are used to produce nonterminal symbols while their sequences become the replacements for those nonterminals. In the last recursion step, the input text $T$ becomes the replacement for the start symbol of the grammar.

\paragraph{Level-Order Unary Degree Sequence (LOUDS).} LOUDS \cite{j89} is a succinct representation that encodes an ordinal tree $T$ with $t$ nodes into a bitmap $B[1..2t+1]$, by traversing its nodes in levelwise order and writing down its arities in unary. The nodes are identified by the position where their description start in $B$. Adding $o(t)$ bits on top of $B$ enables constant-time operations like $\textsf{parent}(u)$ (the parent of node $u$), $\textsf{child}(u,i)$ (the $i$-th child of $u$), $\textsf{psibling}(u)$ (the sibling preceding $u$), $\textsf{nodemap}(u)$ (the level-wise rank of node $u$), $\mathsf{leafrank}(u)$ (the number of leaves in level-order up to leaf $u$), $\mathsf{internalrank}(u)$ (the rank of the internal node $u$ in level-order), and $\mathsf{internalselect}(r)$ (the identifier of the \emph{r-th} internal node in level order).

\Section{Building the grammar}

Let $\mathcal{G}=\{V, \Sigma, \mathqhv{S}, \mathcal{R}\}$ be a context free grammar built from $T$ and that only produces strings in $\mathcal{T}$. $V$ is the set of nonterminals, $\Sigma$ is the alphabet of terminals, $\mathqhv{S}$ is the start symbol and $\mathcal{R}$ is the set of rules. Additionally, denoted the number of rules as $r=|\mathcal{R}|$. The grammar size $g$ is defined as the sum of the lengths of the right-hand sides of $\mathcal{R}$. We refer to the string $C$ in the right-hand side of the rule of $\mathqhv{S}$ as the \emph{compressed string} of $\mathcal{G}$, and its size is denoted as $c=|C|$.

We propose an iterative algorithm called $\mathsf{LMSg}$ for producing $\mathcal{G}$. In every step $i$, we partition the input text $T^{i}$ ($T^{1}=T$) and create a dictionary $\mathcal{D}^{i}$ with all the distinct $LMS$ substrings. Then, for every $F \in \mathcal{D}^{i}$, we create a new rule $\mathqhv{X} \rightarrow F$, where $\mathqhv{X}=p + o$ is the number of rules built before step $i$ and $o$ is the order of $F$ in $\mathcal{D}^{i}$. Finally, we create another text $T^{i+1}$ in which we replace the phrases with their nonterminal symbols, and if there is at least one symbol repeated in $T^{i+1}$, then we perform another iteration $i+1$ using $T^{i+1}$ as input. The algorithm ends when no more new phrases are created from the input text. This procedure is similar to that of Nunes \emph{et al.}~\cite{n2018gr}. Still, we go further and try to reduce the grammar size without losing information for inferring the eBWT of $\mathcal{T}$.

By using ISS to build $\mathcal{G}$ we can determine the relative order of the suffixes of $T$ prefixed by $LMS$ substrings just by looking at their nonterminal symbols. This idea is formally stated with the following lemma:

\begin{lemma}\label{lem:gram}
For two different nonterminals $\mathqhv{X},\mathqhv{Y} \in V$ produced in the same iteration of $\mathsf{LMSg}$, if $\mathqhv{X}<\mathqhv{Y}$, then the suffixes of $T$ whose prefixes are compressed as $\mathqhv{X}$ are lexicographically smaller than the suffixes whose prefixes are compressed as $\mathqhv{Y}$.
\end{lemma}

The only problem is that the occurrences of the phrases in $T^i$ overlap by one character, and that produces redundancy in $\mathcal{R}$. We solve it by discarding the first symbol of every $LMS$ substring. 

\begin{lemma}\label{lem:rank}
The suffix $F[j..|F|]$ of an $LMS$ substrings $F$ can still be used to get the lexicographical rank of a suffix in $T^{i}$ prefixed with it as long as  $|F|-j+1>1$.
\end{lemma}

\begin{proof}
Let a string $D'$ over the alphabet $[\sfnum{0},\sfnum{1}]$ be the \emph{description} of an $LMS$ string $F$. If $F[j]$ is \emph{L-type}, then $D'[j]=\sfnum{1}$ and if $F[j]$ is \emph{S-type} or \emph{LMS-type}, then $D'[j]=\sfnum{0}$. Now consider the set $\mathcal{U}$ with the descriptions of all the phrases of $\mathcal{D}^{i}$. As the pattern $\emph{LS}=\sfnum{10}$ only appears as a suffix in the descriptions, $\mathcal{U}$ is a prefix-free set. Therefore, if an $LMS$ string is a prefix of another $LMS$ substring, then we can decide their relative orders by looking at their descriptions as explained in~\cite{n2009li}.  
\end{proof}

We also ensure that no $\mathqhv{X} \in V$ recursively expands to the suffix-prefix concatenation of two or more strings of $\mathcal{T}$. We call this property \emph{string independence} of $\mathcal{G}$. To guarantee it, the string partition must also be independent.

\begin{definition}\label{def:strind}
The partition of $T^{i}$ in the i-th iteration of $\mathsf{LMSg}$ is string-independent iff the recursive expansion of every symbol $T^{i}[u]$ spans at most one string $T_j \in \mathcal{T}$.
\end{definition}

During the execution of $\mathsf{LMSg}$, we ensure the string independence of the partition by cutting each $LMS$ substring $F$ in $p>1$ segments if the symbols of $F$ cover $p$ different strings of $\mathcal{T}$. We call those segments where the last character recursively expands to a suffix of some $T_j \in \mathcal{T}$ \emph{suffix phrases}. Although the suffix phrases do not meet Lemma~\ref{lem:rank} (they do not necessarily have the $LS$ suffix in their descriptions), it is still possible to assign them unique lexicographical ranks.

\begin{lemma}\label{lem:suffs}
A suffix phrase $S$ generated in a string independent partition cannot be a prefix of any other string in $\mathcal{D}^{i}$.
\end{lemma}

\begin{proof}
Assume $F$ is a prefix of another string $F'$. Also assume that $F$ recursively expands to the substring $A\texttt{\$}$ of $T$. As $F$ is a prefix, it means that $F'$ must expand to a substring $A\texttt{\$}B$, with $B$ also being a substring of $T$. This implication contradicts the definition of string independence as $F'$ spans two consecutive strings.
\end{proof}

\paragraph{Reducing the number of nonterminals.} 
We discard the phrases that are not useful for either compressing or producing the eBWT of $\mathcal{T}$. The symbols in these phrases are \emph{transferred} to subsequent iterations of $\mathsf{LMSg}$ hoping they will be encapsulated within more useful contexts. We do not consider a substring as a phrase for $\mathcal{D}^{i}$ in two cases; (i) all its symbols appear once in $T^{i}$ or (ii) its length is less than two.

With this modification, the partition of $T^{i}$ now yields two sets, $\mathcal{D}^{i}$ and a set $I^{i}$ with the symbols of $T^{i}$ to be transferred. If $\mathcal{D}^{i}$ is empty, then we stop $\mathsf{LMSg}$ and return $\mathcal{G}$. If not, then we sort the phrases of $I^{i} \cup \mathcal{D}^{i}$ in lexicographical order. Once we finish, we update the left-hand side in $\mathcal{R}$ of every nonterminal $s \in I^{i}$ to $p+r'$, where $p$ is the size of $\mathcal{R}$ before iteration $i$ and $r'$ is the rank of $s$ in $I^{i} \cup \mathcal{D}^{i}$. We also update the previous references to $s$ in the right-hand sides of $\mathcal{R}$. For the phrases in $\mathcal{D}^{i}$, we create new nonterminals using their ranks in $I^{i} \cup \mathcal{D}^{i}$.

\paragraph{Reducing the grammar size.}  We scan $\mathcal{R}$ and change every left-hand character with the smallest unused symbol (the nonterminals produced by $\mathsf{LMSg}$ are non-consecutive due to the transfer of symbols). As we do the replacement, we keep track of the changes so we can update the references of the characters in the right-hand sides of $\mathcal{R}$. Once the grammar is collapsed, we recursively create new rules from the suffixes of size two that appear in more than one distinct right-hand side, and we stop when all such suffixes are unique. We refer to these new nonterminals as SP (suffix pairing). It might happen that the complete sequence $F$ of an $\mathsf{LMSg}$ rule $\mathsf{X} \rightarrow F$ appears as a proper suffix in one or more right-hand sides. In such situation, we do not create a new rule but reuse the value of $\mathsf{X}$ to replace those proper suffixes. When this happens, we consider $\mathqhv{X}$ to have a \emph{dual} context as it occurs as an $\mathsf{LMSg}$ nonterminal but also as an SP nonterminal.

\paragraph{Encoding the grammar.} 

We use the \emph{grammar tree} data structure proposed by Claude \emph{et al}.~\cite{cn2012im} (denoted here as \gt{}) to store $\mathcal{G}$. We make, however, some modifications to compute the eBWT of $\mathcal{T}$ in a more efficient way. The procedure is as follows; we create a root node labeled with the start symbol $\mathqhv{S}$ and with $c$ children, one for every symbol in the right-hand of its rule in $\mathcal{R}$. Then, we create the nodes in the subtrees of the root by visiting in level-order the rules of the nonterminals to which $C$ recursively expands. During this process, when we reach a rule $\mathqhv{X} \rightarrow F$ for the first time, we create a new internal node $v$ with $|F|$ children and labeled with $x+\sigma$, where $x$ is the number of internal nodes in level-order up to $v$. Nevertheless, if the symbol has a dual context, then we create $v$ only if the occurrence of $\mathqhv{X}$ in the visit corresponds to an $\mathsf{LMSg}$ nonterminal. When this is not the case, we create a leaf $v'$ with an empty label instead. We refer to $x+\sigma$ as the \emph{identifier} of $\mathqhv{X}$ in \gt{}. The next time we reach this rule in the traversal, we create a leaf $v'$ labeled with $x+\sigma$. In the case the identifier is still unknown (i.e., $\mathqhv{X}$ has dual context and all the occurrences we have visited so far are SP), we leave the label of $v'$ empty. Later, when we reach the first occurrence of $\mathqhv{X}$ as $\mathsf{LMSg}$, we create a new internal node $v$ and label with $x+\sigma$ all the empty leaves that should point to this identifier. Finally, when we visit a terminal symbol, we create a leaf labeled with its value. We encode the topology of the resulting tree in a bitmap $K$ using LOUDS. Additionally, the leaf labels are stored in a vector $Z$ using Canonical Huffman codes~\cite{sch964gen}.

For simulating in \gt{} a traversal of the parse tree of $\mathcal{G}$ we use the constant-time navigational functions $\mathsf{child}$ and $\mathsf{parent}$ defined for LOUDS, but also an extra function $\mathsf{label}(v)$ that returns the label of a node $v$. If $v$ is a leaf, then the function returns $Z[\mathsf{leafrank}(v)]$. On the other hand, if $v$ is an internal node, then it returns $\mathsf{internalrank}(v)+\sigma$. When we reach a leaf $u$, if $\mathsf{label}(u) \leq \sigma$, then we stop descending as we reach a terminal symbol. If that is not the case, then we continue the traversal from the subtree rooted at $v=\mathsf{internalselect}(K, \mathsf{label}(v)-\sigma)$. 

\Section{Building the eBWT from the grammar}\label{sec:lpg2bwt}

Our framework for building the eBWT of $\mathcal{T}$ consists of two algorithms, $\mathsf{GLex}$ and $\mathsf{infBWT}$. The first one computes the original lexicographical ranks of the nonterminals generated by $\mathsf{LMSg}$ and the second uses these ranks to produce the eBWT. 

\paragraph{Computing the ranks of the nonterminals.} $\mathsf{GLex}$ is an iterative method that reconstructs the steps of $\mathsf{LMSg}$. In every iteration, the algorithm produces a set $L^{i} \in [1..r+\sigma]$ with the identifiers in \gt{} for the phrases in $\mathcal{D}^{i}$. Then, it computes another set $R^{i}$ with the lexicographical ranks of these phrases. Finally, it creates a function $f^{i}: L^{i} \rightarrow R^{i}$ that maps the identifier $l \in L^{i}$ of a phrase in $\mathcal{D}^{i}$ with its lexicographical rank. The result of $\mathsf{GLex}$ is a set of $h$ distinct triplets $(L^{i}, R^{i}, f^{i})$, where $h$ is the number of iterations of $\mathsf{LMSg}$.

For computing $L^{i}$, we visit the internal nodes of \gt{} in level-order and check which of them encode phrases of $\mathcal{D}^{i}$. For this task we use the following lemma:

\begin{lemma}\label{lem:stair}
Let $\mathqhv{X}\rightarrow F \in \mathcal{R}$ be a nonterminal rule generated by $\mathsf{LMSg}$. If all the suffixes of $F$ up to position $1<k\leq|F|-1$ appear in more than one right-hand side in $\mathcal{R}$, then after reducing the grammar size, every subtree rooted at some node labeled with $\mathqhv{X}$ in the parse tree will have the original last $|F|-k+1$ children of $\mathqhv{X}$ recursively encapsulated from right to left inside new internal nodes.
\end{lemma}

\begin{proof}
Consider a node $v$ in the parse tree of $\mathcal{G}$ that represents the occurrence of an $\mathsf{LMSg}$ nonterminal. After reducing the grammar size, its subtree adopts a stair-like shape as the SP rules are recursively built from right to left. 
\end{proof}

By using the stair-like pattern described in Lemma~\ref{lem:stair}, we can recognize occurrences of $\mathsf{LMSg}$ nonterminals just by looking at the topology of the parse tree of $\mathcal{G}$.

\begin{lemma}\label{lem:lmsnode}
A node $v$ of \gt{} encodes the occurrence of a nonterminal produced in the iteration $i$ of $\mathsf{LMSg}$ if the leftmost child of $v$ in the parse tree is labeled with a symbol $l \in L^{i-1}$ and either $v$ if the leftmost child of its parent or the left sibling of $v$ is labeled with a symbol $l' \notin L^{i-1}$.
\end{lemma}

\begin{proof}
A nonterminal $v$ whose first child has a label $l \in L^{i-1}$ is either an $\mathsf{LMSg}$ nonterminal of the iteration $i$ or an SP nonterminal. If it is SP, then, due to the stair-like pattern, the label of its left sibling must be in $L^{i-1}$, otherwise $v$ is $\mathsf{LMSg}$.
\end{proof}

Once we compute the symbols in $L^{i}$, we decompress and sort their associated phrases to generate $R^{i}$. For that end, we regard $L^{i-1}$ as a set of logical leaves in the parse tree of $\mathcal{G}$. Thus, if during the decompression of an internal node $v=\mathsf{internalselect}(l-\sigma)$, with $l \in L^{i}$, we reach a node $v'$ with $\mathsf{label}(v') \in L^{i-1}$, then we do not visit its subtree but spell its symbol $f^{i-1}(\mathsf{label}(v')) \in R^{i-1}$. After decompressing $\mathcal{D}^{i}$, we apply the same string sorting mechanism of $\mathsf{LMSg}$. The function $f^{i}$ is implemented by encoding $L^{i}$ as a bitmap $L[1..r+\sigma]$ where $L[l]$ is set to $\sfnum{1}$ if $l \in L^{i}$ and $\sfnum{0}$ otherwise. Additionally, we augment $L$ with constant-time $\mathsf{rank}$ support \cite{Cla96,navarro2016compact}, so that $\mathsf{rank}(L,l)$ is the number of $\sfnum{1}$s in $L[1..l]$, and we store at position $R^{i}[\mathsf{rank}(L,l)]$ the lexicographical rank associated to $l$. Finally, we pass the triplet $(L^{i}, R^{i}, f^{i})$ to the next iteration $i+1$ to compute $L^{i+1}$, $R^{i+1}$ and $f^{i+1}$.

\paragraph{Inferring the eBWT.} $\mathsf{infBWT}$ is also an iterative process of $h$ steps. In the first one, we produce the eBWT $B^{h}$ of the compressed string $C$ by sorting the nodes at depth one of \gt{} in (circular) lexicographical order. For this task, we insert them in an array $A[1..c]$ such that if a node $v$ has order $o$ in $R^{h}$, then we store it in the bucket $o$ of $A$. Subsequently, we sort the buckets of $A$ independently. During the process, if two nodes have the same rank in $R^{h}$, then we walk through their right siblings until finding nodes with different labels. It might happen that one of the siblings we reach in the walk represents a suffix of a string in $\mathcal{T}$. In such case we move backward until finding the first left sibling encoding the suffix of another string. The idea is to simulate the circularity of the elements in $\mathcal{T}$. Once we finish the sorting, we insert in $B^{h}$ the orders in $R^{h}$ of the (circular) left siblings of the nodes in $A$.  

In the rest of the iterations, we receive as inputs the eBWT $B^{h-i+1}$ of $T^{h-i+1}$ (the temporary string of $\mathsf{LMSg}$) and the triplet $(L^{h-i}, R^{h-i}, f^{h-i})$. We scan $B^{h-i+1}$ from left to right to decompress the occurrences of the phrases in $\mathcal{D}^{h-i}$. As we spell a phrase $F \in \mathcal{D}^{h-i}$ from some position $B^{h-i+1}[j]$, we push every possible pair $(F[k],S)$, with $k \in [1..|F|-1]$, into a semi-external vector $\mathcal{Q}$, where $S$ is a proper suffix $F[j+1..|F|]$ of size at least two. To reduce the space usage, we store $S$ using a specific node $v$ of \gt{} from which we can decompress it later (an SP node in most of the cases). When $S=F[|F|]$, we first obtain the right context symbol $s$ in $T^{h-i+1}$ of $B^{h-i+1}[j]$ . Then, we associate $(F[|F|],s)$ to a new identifier $q>g$ and push $(F[|F|-1], q)$ to $\mathcal{Q}$.

After scanning $B^{h-i+1}$, we sort the \emph{distinct} right elements of $\mathcal{Q}$ by decompressing them from \gt{}. Once we finish, we rearrange $\mathcal{Q}$ according to the resulting ranks, and without changing the relative order of the elements with the same value. We thus extract $B^{h-i}$ by concatenating the left symbols of $\mathcal{Q}$. In the last step of $\mathsf{infBWT}$, the resulting $B^{h-i}$ is in fact the eBWT of $\mathcal{T}$.

\Section{Experiments}

We implemented our framework as a tool called LPG (\url{https://bitbucket.org/DiegoDiazDominguez/lms_grammar/src/bwt_imp2}). The software is written in \texttt{C++} and uses the \texttt{SDSL-lite} library~\cite{gbmp2014sea}. We compared the performance of LPG against BigRepair~\cite{g2019rpair}, 7-zip and the FM-index~\cite{ferragina2005indexing}. BigRepair (BR) is a space-efficient variation of RePair for large repetitive collections. We encoded the BigRepair grammars with the recent representation of Gagie \emph{et al.}~\cite{gag2020prac}, which allows fast random accession to substrings of the text. For the FM-Index, we consider both the regular version (FM) and the Run-Length compressed version (RLFM). The BWTs for the FM-indexes were calculated using egap~\cite{eg2019ext}. When parallelization was possible, we ran the experiments with 10 threads.


We used as input five distinct collections of reads produced from different human individuals. This data was obtained from the Human Genome Diversity Project. The datasets were identified with the number of individuals they contained. Their sizes in GB were 1=12.77, 2=23.43, 3=34.30, 4=45.89 and 5=57.37. All the reads were 152 character long and had an alphabet of six symbols (\texttt{A,C,G,T,N,\$}). The instance of BR with collection 5 returned an error so it was not included in the analyses. For dataset 1, we allowed BR to use at most 72GB (6x the input size) of working memory. However, with the rest of the collections we had to increase that value to 275,36 GB as the program was taking too long to finish. The performance of the compressors is shown in Figure 1. 

We measured the time for randomly accessing the reads from the compressed representations. To support fast accession in the FM-indexes, we sampled reads in the text at regular intervals. For every sampled element, we stored the BWT position of its last character. The sampling rate for RLFM was 0.05 while for FM was 1. In addition, we augmented the LPG instances with a bitmap $B[1..c]$ that mark in \gt{} the nodes at depth one that recursively expand to string suffixes. We also encoded the leaf labels of \gt{} using arrays of $\log r$-bit cells. The results are depicted in the left side of Table 1.~We implemented $\mathsf{GLex}$ and measured its space and time consumption. The results are shown on the right side of Table 1. All the experiments were carried out on a machine with Debian 4.9, 736 GB of RAM and processor Intel(R) Xeon(R) Silver @ 2.10GHz, with 32 cores.

\begin{figure}[t]\label{fig:perf}
\centering
\includegraphics[width=0.9\textwidth]{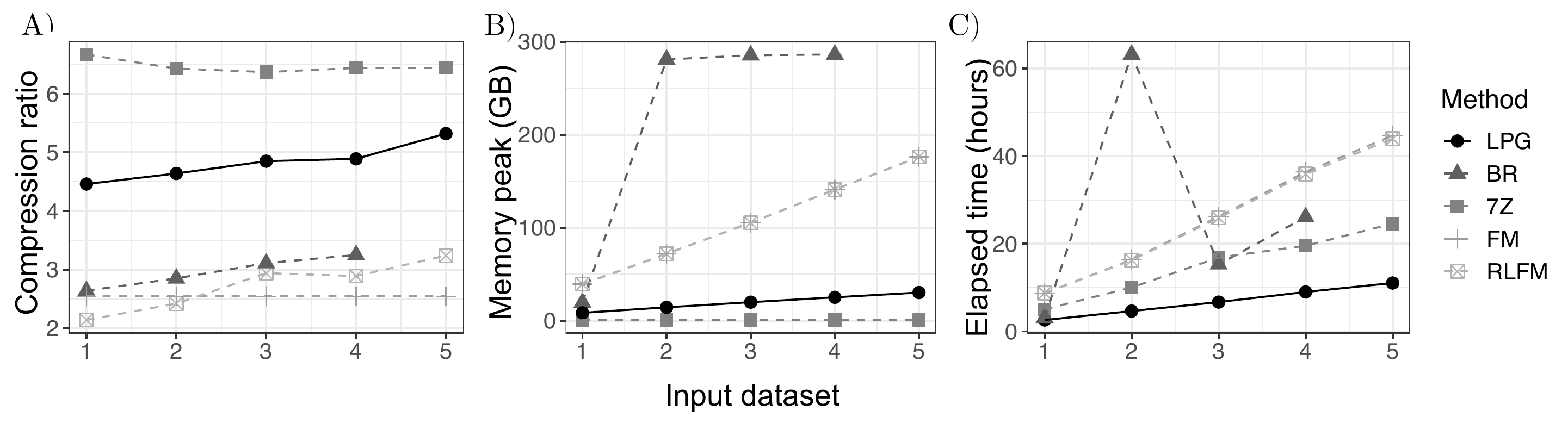}
\caption{Performance of the different compressors. The compression ratio is measured as the size of the plain text divided by the size of the final compressed representation.}
\end{figure}

\begin{table}[ht]
\centering
 \resizebox{0.75\textwidth}{!}{%
  \begin{tabular}{ccccc c c}
  \toprule
    & \multicolumn{4}{c}{Random access} &\multicolumn{2}{c}{$\mathsf{GLex}$} \\
    \cmidrule(lr){2-5}
    \cmidrule(lr){6-7}
    Input &  LPG & BR & RLFM & FM& Elap. time & Space usage \\
    \midrule
    1     &  104.30  & 98.67  & 6,699.06 & 90.40 & 0.25 &  0.49\\
    2     &  111.35  & 101.59 & 6,694.38 & 101.78 & 0.12 &  0.45\\
    3     &  124.04  & 98.56  & 7,422.91 & 109.68 & 0.08 &  0.42\\
    4     &  128.58  & 104.72 & 7,280.12 & 113.31 & 0.07 &  0.39\\
    \bottomrule
    \end{tabular}}
    \caption{LPG experiments. The left side depicts the average time in $\mu$secs to randomly accessing a read. The right side shows the running time and space usage of $\mathsf{GLex}$. The time is expressed as $\mu$secs per byte and the space as the fraction of the uncompressed input.}
    \label{table:tab}
\end{table}

\paragraph{Results and discussion.} The average compression ratio of LPG was 4.83. This result was better than the one obtained by BR and RLFM (2.96 and 2.73, respectively), but worse than that of 7Z (6.47). Although 7Z outperformed the other methods at reducing the space, the difference was reduced as the inputs grew and became more repetitive. For instance, the gap in the compression ratio between 7Z and LPG for collection 1 was 2.24, while for collection 5 was 1.12. The poor performance of BR may be due to its prepossessing step (Prefix-Free Parsing) did not capture well the repetitiveness in the reads. BR produced, on average, 322 million more grammar rules than LPG. On the other hand, the small compression ratios obtained by RLFM can be due to the number of BWT runs in our inputs was not as small as in other text families. The run heads represented, on average, 23\% of our inputs. Regarding the memory peaks, the consumption of 7Z was negligible (0.7 GB). In contrast, LPG required a much more considerable amount of working space (about 58\% of the input size). Still, this value was far less than that of BR and RLBWT, that used 7 and 3 times the input size, respectively. In elapsed time, LPG outperformed all the other methods. The instance of BR with collection 2 took much more time compared to collection 3 and 4 (63.18 hours versus 15.31 and 26.08 hours, respectively). We assume this behaviour is a bug in the implementation. The performance for randomly accessing the reads was similar between LPG and BR and FM (between 90 and 128 $\mu$secs), and slow for RLBWT (mainly because of the small sampling). Still, in all the cases the mean space overhead over the compressed representation was small (16\% for LPG, 7\% for FM and less than 1\% for RLFM). Finally, $\mathsf{GLex}$ required about 0.26 $\mu$secs per input byte and used an amount of working space proportional to half the space of the uncompressed collection.

\Section{References}
\bibliographystyle{IEEEbib}

\end{document}